\documentclass[11pt,a4paper]{article}

\usepackage{amsmath,amssymb,amsthm,mathrsfs}
\usepackage{enumitem}
\usepackage{booktabs}
\usepackage{float}
\usepackage{microtype}
\usepackage{hyperref}
\hypersetup{
  hidelinks,
  pdftitle={Statistical Inference for Probability Barycenters and Kolmogorov Moments},
  pdfauthor={Manuela-Simona Cojocea}
}

\newtheorem{theorem}{Theorem}[section]
\newtheorem{proposition}[theorem]{Proposition}
\newtheorem{lemma}[theorem]{Lemma}
\newtheorem{corollary}[theorem]{Corollary}
\theoremstyle{definition}
\newtheorem{definition}[theorem]{Definition}
\newtheorem{assumption}[theorem]{Assumption}
\newtheorem{remark}[theorem]{Remark}

\newcommand{\E}{\mathbb{E}}
\newcommand{\Prob}{\mathbb{P}}
\newcommand{\R}{\mathbb{R}}
\newcommand{\Var}{\operatorname{Var}}
\newcommand{\Cov}{\operatorname{Cov}}

\newcommand{\IF}{\operatorname{IF}}
\newcommand{\diag}{\operatorname{diag}}

\title{Statistical Inference for Probability Barycenters and Kolmogorov Moments}
\author{Manuela-Simona Cojocea}
\date{}

\begin{document}

\maketitle

\begin{abstract}
Let $G:I\to(0,1)$ be a continuous strictly increasing bijection, called a probability coordinate chart, and represent an observation $X$ by its probability coordinate $U=G(X)$. Averaging is always well defined on the bounded coordinate scale. For barycentric inference, the first coordinate moment is interpreted through $G^{-1}$. Higher initial coordinate moments similarly generate initial Kolmogorov moments by pullback, while centred coordinate moments remain on the probability scale. In each case the inferential law depends on how the chart entered the problem. This paper develops inference for probability barycenters, initial Kolmogorov moments, and centred coordinate moments under fixed, intrinsic, and estimated charts. The geometric construction and its probability theory are developed in \cite{Cojocea2026}. Here the emphasis is deliberately practical: variance estimation, confidence sets, and bootstrap approximation.

For a fixed benchmark chart, the coordinate parameter $p_G=\E[G(X)]$ is treated as the primary inferential object. Confidence intervals are constructed on the probability scale and then pulled back through $G^{-1}$, which preserves the geometry and allows asymmetry in value space. Hoeffding's inequality also yields finite-sample distribution-free intervals. The intrinsic case requires more care. The empirical CDF is not an admissible chart, and its natural generalised plug-in construction collapses exactly to a middle order statistic. Feasible intrinsic inference reduces to median inference. We separate this self-induced functional from the frozen oracle calculation that produces the smaller variance $1/(12f(m)^2)$.

For estimated location--scale charts, the asymptotic expansion contains a calibration correction that records the first-order effect of learning the chart from the same sample. We turn that correction into an implementable influence-function variance estimator and a bootstrap procedure that recalibrates the chart in every resample. Joint covariance theory is then developed for finite vectors of initial and centred coordinate moments, with initial Kolmogorov moments obtained by pulling the former through $G^{-1}$. Bounded coordinates do not make uncertainty disappear. They make its sources visible enough to estimate: variation on the probability scale, amplification by the inverse chart, and the contribution of chart calibration together with its covariance with the coordinate mean.
\end{abstract}

\section{Introduction}

Probability geometry begins with a simple decision: the scale on which averaging will be performed. Let $I\subseteq\R$ be an interval and let
\[
G:I\longrightarrow(0,1)
\]
be a continuous strictly increasing bijection. In the framework introduced in \cite{Cojocea2026}, $G$ is called a \emph{probability coordinate chart}. Because $G$ is strictly increasing and bijective,
\[
x<y \quad\Longleftrightarrow\quad G(x)<G(y).
\]
Equivalently, $G$ is an order-isomorphism between the observation scale and the probability scale. It preserves the ordering of observations while generally changing the distances between them. Once $G$ is extended by $0$ to the left of $I$ and by $1$ to the right, it also becomes the cumulative distribution function of the reference law generated by that chart. The CDF interpretation is earned by the coordinate structure itself. It need not be imposed beforehand.

For a random variable $X$ supported on $I$, write
\[
U=G(X),\qquad p_G=\E[U],\qquad b_G=G^{-1}(p_G).
\]
Reference~\cite{Cojocea2026} develops the geometry associated with this construction and proves the probability results needed here, including existence, limit theory, concentration, and moment determinacy. The present paper takes those results as established. Its concern is the uncertainty left behind when the population quantities are replaced by data.

Arithmetic averaging has not disappeared. It takes place in the chart selected for the problem. This small shift of viewpoint creates several inferential obligations. The coordinate mean has a sampling variance that must be estimated. A confidence interval formed on $(0,1)$ must be carried back to the observation scale without forgetting the curvature of $G^{-1}$. When the chart is learned from the observations, the calibration step contributes its own first-order randomness. Initial coordinate moments, their Kolmogorov pullbacks, and centred coordinate moments also need a joint treatment because their estimators are strongly dependent functions of the same transformed sample.

The status of the chart decides which of these obligations arise. A fixed benchmark chart is declared by the analyst and gives a perfectly feasible target, although that target remains anchored to the chosen geometry. The intrinsic chart $F_P$ has a special conceptual position, but it is unknown. Its empirical CDF surrogate falls outside the admissible chart class and collapses to a middle order statistic. Estimated charts occupy the useful space between these cases. Their shape is declared, while location and scale are calibrated from the data. Equivariance is recovered, and the first-order effect of that calibration appears explicitly in the influence function.

The same inferential viewpoint is applied to the finite vector
\[
\bigl(G(X),G(X)^2,\ldots,G(X)^q\bigr).
\]
Working with this vector reveals the covariance between moment orders before any pullback is performed. Initial coordinate moments, their pulled-back Kolmogorov counterparts, centred coordinate moments, simultaneous intervals, and joint bootstrap procedures can then be obtained from one coherent calculation.

Section~\ref{sec:targets} introduces the targets and the three chart regimes. Sections~\ref{sec:fixed}--\ref{sec:estimated} develop inference for fixed, intrinsic, and estimated charts. Section~\ref{sec:moments} treats initial coordinate moments, initial Kolmogorov moments, and centred coordinate moments. Robustness and efficiency are examined in Section~\ref{sec:robustness}. Section~\ref{sec:numerics} gives asymptotic variance benchmarks and a reproducible Monte Carlo illustration.

\section{Statistical targets and chart regimes}\label{sec:targets}

\subsection{Probability barycenters}

Let $G:I\to(0,1)$ be a probability coordinate chart and let $X$ satisfy $\Prob(X\in I)=1$.

\begin{definition}[Coordinate mean and probability barycenter]
The \emph{coordinate mean} and its pullback are
\[
p_G(P)=\E_P[G(X)],
\qquad
b_G(P)=G^{-1}\!\bigl(p_G(P)\bigr).
\]
We call $b_G(P)$ the \emph{probability barycenter} associated with $G$.
\end{definition}

The pair $(p_G,b_G)$ carries the whole inferential logic of the paper. Sampling takes place through the bounded variable $G(X)$, so $p_G$ is the natural object for variance estimation and concentration. The pullback $b_G$ is the answer expressed on the original scale. Its uncertainty inherits the local steepness of $G^{-1}$. Losing sight of this separation makes the method look simpler than it is and hides precisely where geometric amplification enters.

\subsection{Initial Kolmogorov moments and centred coordinate moments}

Let $U=G(X)$. For $r\geq1$, define
\[
p_r^{(G)}=\E[U^r],
\qquad
M_r^{(G)}=G^{-1}\!\bigl(p_r^{(G)}\bigr).
\]
The quantities $p_r^{(G)}$ are the initial coordinate moments. Their pullbacks $M_r^{(G)}$ are the initial Kolmogorov moments. The centred coordinate moments are
\[
\mu_{r,c}^{(G)}=\E\bigl[(U-p_1^{(G)})^r\bigr].
\]
A centred coordinate moment describes spread or shape around the coordinate mean. It is not itself a probability level and may be zero or negative, so applying $G^{-1}$ would have no coherent geometric meaning.

\subsection{Three chart regimes}

Three regimes recur throughout the paper. Under a \emph{fixed chart}, $G$ is completely specified before the observations are seen. In the \emph{intrinsic construction}, the population chart is $F_P$, and feasibility is attempted through the empirical CDF and its generalised inverse. Under an \emph{estimated chart}, $G_\theta$ belongs to a finite-dimensional family and the parameter $\theta$ is learned from the sample. These labels identify genuinely different statistical objects. Their random terms are different, so their variance formulas cannot be exchanged by analogy.

\section{Inference with a fixed benchmark chart}\label{sec:fixed}

Let $X_1,\ldots,X_n$ be i.i.d. from $P$, let $G$ be fixed, and set
\[
U_i=G(X_i),\qquad
\widehat p_G=\overline U_n=\frac1n\sum_{i=1}^nU_i,
\qquad
\widehat b_G=G^{-1}(\widehat p_G).
\]
The required limit theory is established in \cite{Cojocea2026}. Here the question is more demanding in practice: how can the asymptotic variance be estimated from the sample, and how should uncertainty be reported after the coordinate mean is pulled back through a nonlinear inverse chart?

\subsection{Variance estimation}

Let
\[
v_G=\Var\bigl(G(X)\bigr),
\qquad
\tau_G^2=\frac{v_G}{[G'(b_G)]^2},
\]
and define
\[
\widehat v_G
=
\frac1{n-1}\sum_{i=1}^n(U_i-\overline U_n)^2,
\qquad
\widehat\tau_G^2
=
\frac{\widehat v_G}{[G'(\widehat b_G)]^2}.
\]

\begin{proposition}[Consistent fixed-chart variance estimator]\label{prop:fixed-var}
Assume that $G$ is continuously differentiable on a neighbourhood of $b_G$ and that $G'(b_G)>0$. Then
\[
\widehat\tau_G^2\xrightarrow{\Prob}\tau_G^2.
\]
Thus $\widehat\tau_G^2/n$ consistently estimates the first-order sampling variance of $\widehat b_G$.
\end{proposition}

\begin{proof}
Because $U_i\in(0,1)$, the sample variance satisfies $\widehat v_G\to v_G$ almost surely. The strong consistency result in \cite{Cojocea2026} gives $\widehat b_G\to b_G$ almost surely. Continuity and positivity of $G'$ at $b_G$ imply $G'(\widehat b_G)\to G'(b_G)$, and the result follows by the continuous mapping theorem.
\end{proof}

\subsection{Asymptotic confidence intervals}

The native confidence interval lives on the probability scale. Constructing it there keeps the bounded parameter visible and lets monotonicity carry the entire set back to $I$. A symmetric interval around $\widehat b_G$ treats the inverse chart as locally straight. That approximation is often adequate, but it should not be mistaken for the geometry itself.

\begin{theorem}[Coordinate-scale Wald inference]\label{thm:fixed-wald}
Under the assumptions of Proposition~\ref{prop:fixed-var}, let $z_{1-\alpha/2}$ denote the $(1-\alpha/2)$ standard normal quantile and define
\[
I_{p,n}(\alpha)
=
\left[
\widehat p_G-z_{1-\alpha/2}\sqrt{\frac{\widehat v_G}{n}},
\widehat p_G+z_{1-\alpha/2}\sqrt{\frac{\widehat v_G}{n}}
\right]\cap[0,1].
\]
Then
\[
\Prob\bigl(p_G\in I_{p,n}(\alpha)\bigr)\longrightarrow1-\alpha.
\]
Using the extended inverse convention
\[
G^{-1}(0)=\inf I,
\qquad
G^{-1}(1)=\sup I,
\]
\[
I_{b,n}(\alpha)=G^{-1}\!\bigl(I_{p,n}(\alpha)\bigr)
\]
satisfies
\[
\Prob\bigl(b_G\in I_{b,n}(\alpha)\bigr)\longrightarrow1-\alpha.
\]
\end{theorem}

\begin{proof}
The first assertion is the studentised central limit theorem for the bounded variables $U_i$. Since $G^{-1}$ is increasing, $p_G\in I_{p,n}(\alpha)$ if and only if $b_G\in G^{-1}(I_{p,n}(\alpha))$.
\end{proof}

\begin{remark}[Direct value-space Wald interval]
The usual delta-method interval
\[
\widehat b_G
\pm
z_{1-\alpha/2}\frac{\widehat\tau_G}{\sqrt n}
\]
is asymptotically valid. The transformed interval of Theorem~\ref{thm:fixed-wald} is preferable when $G^{-1}$ is noticeably curved, because it respects the probability boundaries and allows the corresponding value-space interval to be asymmetric.
\end{remark}

\subsection{Finite-sample distribution-free intervals}

Hoeffding's inequality applied to the bounded coordinates gives a finite-sample confidence set without estimating a variance.

\begin{corollary}[Hoeffding confidence interval]\label{cor:hoeffding-ci}
For $\alpha\in(0,1)$ let
\[
\varepsilon_{n,\alpha}
=
\sqrt{\frac{\log(2/\alpha)}{2n}}.
\]
Then
\[
\Prob\left(
|\widehat p_G-p_G|\leq\varepsilon_{n,\alpha}
\right)\geq1-\alpha.
\]
By monotonicity,
\[
C_{b,n}(\alpha)
=
G^{-1}\!\left(
[\widehat p_G-\varepsilon_{n,\alpha},
 \widehat p_G+\varepsilon_{n,\alpha}]\cap[0,1]
\right)
\]
is a finite-sample confidence interval for $b_G$ with coverage at least $1-\alpha$.
\end{corollary}

\begin{remark}
The two intervals answer slightly different inferential needs. The Hoeffding interval buys finite-sample coverage without learning anything about the distribution of $U$. It pays for that guarantee with width. The Wald interval uses the observed coordinate variance and is usually shorter. Reporting both makes the bargain visible instead of hiding it behind a single standard error.
\end{remark}

\subsection{Bootstrap inference}

Let $X_1^*,\ldots,X_n^*$ be an ordinary nonparametric bootstrap sample drawn from the empirical distribution, and define
\[
\widehat b_G^*
=
G^{-1}\!\left(\frac1n\sum_{i=1}^nG(X_i^*)\right).
\]

\begin{theorem}[Bootstrap validity for a fixed chart]\label{thm:fixed-bootstrap}
Under the assumptions of Proposition~\ref{prop:fixed-var}, conditionally on the data,
\[
\sqrt n\,(\widehat b_G^*-\widehat b_G)
\rightsquigarrow
\mathcal N(0,\tau_G^2)
\]
in probability. Hence the ordinary nonparametric bootstrap consistently estimates the first-order distribution of $\sqrt n(\widehat b_G-b_G)$.
\end{theorem}

\begin{proof}
Conditional bootstrap validity holds for the sample mean of the bounded variable $U=G(X)$. The map $u\mapsto G^{-1}(u)$ is continuously differentiable at $p_G$, and the bootstrap delta method carries the approximation to value space.
\end{proof}

\subsection{Anchoring and non-equivariance}

\begin{proposition}[Non-equivariance of fixed-chart barycenters]\label{prop:nonequivariance}
Let
\[
T_G(P)=G^{-1}(\E_P[G(X)]).
\]
In general, $T_G$ is not location--scale equivariant. Specifically, there exist $P$, $a>0$, and $c\in\R$ such that
\[
T_G(P_{a,c})\neq aT_G(P)+c,
\]
where $P_{a,c}$ is the law of $aX+c$ for $X\sim P$.
\end{proposition}

\begin{proof}
Let $G=\Phi$ and $X\sim\mathcal N(\mu,1)$. The identity $\E[\Phi(X)]=\Phi(\mu/\sqrt2)$ gives
\[
T_\Phi(\mathcal N(\mu,1))=\mu/\sqrt2.
\]
Thus the response of the functional to a shift by $c\neq0$ is $c/\sqrt2$.
\end{proof}

\begin{remark}[Inference is conditional on the declared geometry]
This is the point at which seductive variance comparisons must be resisted. A narrow interval for a fixed-chart barycenter is a precise statement about the anchored target $b_G(P)$. It says nothing automatically about the quality of estimation for the arithmetic mean, the median, or an equivariant location parameter. The chart helps determine the sampling variability because it has already helped determine the estimand.
\end{remark}

\section{The intrinsic empirical surrogate}\label{sec:intrinsic}

The intrinsic chart is the most tempting choice in the framework. It is generated by the law of $X$ itself, and at the population level
\[
b_F(P)=F^{-1}(1/2),
\]
so the intrinsic barycenter is the median. Feasibility changes the story. The empirical distribution function
\[
\widehat F_n(x)=\frac1n\sum_{i=1}^n\mathbf 1_{(-\infty,x]}(X_i)
\]
is discontinuous, non-injective, and fails to cover $(0,1)$. It does not belong to the admissible chart class. Its generalised inverse still permits a natural plug-in surrogate,
\[
\widehat b_n^{\mathrm{int}}
=
\widehat F_n^{-1}\!\left(
\frac1n\sum_{i=1}^n\widehat F_n(X_i)
\right).
\]

\subsection{Exact rank collapse}

\begin{lemma}[Rank identity]\label{lem:rank}
If $F$ is continuous, then almost surely
\[
\frac1n\sum_{i=1}^n\widehat F_n(X_i)
=
\frac{n+1}{2n}.
\]
\end{lemma}

\begin{proof}
Almost surely the observations are distinct and $\widehat F_n(X_i)=R_i/n$, where the ranks form a permutation of $1,\ldots,n$. Hence
\[
\frac1n\sum_{i=1}^n\widehat F_n(X_i)
=
\frac1{n^2}\sum_{r=1}^nr
=
\frac{n+1}{2n}.
\]
\end{proof}

\begin{theorem}[Degeneracy of the fully empirical intrinsic surrogate]\label{thm:collapse}
Assume that $F$ is continuous. Then, almost surely,
\[
\widehat b_n^{\mathrm{int}}
=
X_{(\lceil(n+1)/2\rceil)}.
\]
Thus the surrogate is the sample median for odd $n$ and the upper middle order statistic for even $n$.
\end{theorem}

\begin{proof}
By Lemma~\ref{lem:rank}, the argument of the generalised inverse is $(n+1)/(2n)$. Since $\widehat F_n^{-1}(u)=X_{(\lceil nu\rceil)}$, the result follows.
\end{proof}

\begin{remark}[What the rank identity removes]
The collapse is exact and finite-sample. The inner average no longer contains information about the shape of $F$, because the empirical CDF has already replaced the observations by the deterministic set of ranks $1/n,\ldots,n/n$. All remaining randomness enters through the empirical quantile map. The oracle averaging mechanism has disappeared before any asymptotic argument begins.
\end{remark}

\begin{corollary}[Asymptotic distribution]\label{cor:intrinsic-asymptotic}
Suppose $F$ is strictly increasing and has a density $f$ that is continuous and positive at the median $m=F^{-1}(1/2)$. Then
\[
\sqrt n\,(\widehat b_n^{\mathrm{int}}-m)
\rightsquigarrow
\mathcal N\!\left(0,\frac1{4f(m)^2}\right).
\]
\end{corollary}

\subsection{Frozen oracle versus self-induced functional}

The notation $G=F$ can conceal a decisive change of functional. Freezing $F$ while the law is perturbed produces one derivative. Recomputing the CDF after perturbation produces another. The distinction changes the influence curve and the attainable variance.

\begin{definition}[Frozen oracle functional]
For a fixed baseline law $P$ with CDF $F_P$, define
\[
T_P^{\mathrm{fr}}(Q)
=
F_P^{-1}\!\left(\E_Q[F_P(X)]\right).
\]
Only the averaging law $Q$ varies. The chart remains frozen at $F_P$.
\end{definition}

\begin{definition}[Self-induced intrinsic functional]
Define
\[
T^{\mathrm{int}}(Q)
=
F_Q^{-1}\!\left(\E_Q[F_Q(X)]\right)
=
F_Q^{-1}(1/2).
\]
Thus $T^{\mathrm{int}}$ is exactly the median functional.
\end{definition}

\begin{proposition}[Two distinct influence functions]\label{prop:oracle-self-if}
Let $P$ have continuous CDF $F$ and density $f$ positive at its median $m$. Then
\[
\IF(x;T_P^{\mathrm{fr}},P)
=
\frac{F(x)-1/2}{f(m)},
\qquad
\E_P\bigl[\IF(X;T_P^{\mathrm{fr}},P)^2\bigr]
=
\frac1{12f(m)^2},
\]
whereas
\[
\IF(x;T^{\mathrm{int}},P)
=
\frac{1/2-\mathbf 1_{(-\infty,m]}(x)}{f(m)},
\qquad
\E_P\bigl[\IF(X;T^{\mathrm{int}},P)^2\bigr]
=
\frac1{4f(m)^2}.
\]
\end{proposition}

\begin{proof}
The first formula follows by differentiating the frozen-chart functional. The second is the classical influence function of the median. Since $F(X)\sim\mathrm{Unif}(0,1)$, the frozen-oracle variance is\par\smallskip
\[
\frac{\Var(F(X))}{f(m)^2}=\frac1{12f(m)^2}.
\]
The median influence function has constant squared magnitude $1/(4f(m)^2)$ almost surely.
\end{proof}

\begin{remark}[Meaning of the oracle variance]
The number $1/(12f(m)^2)$ is mathematically real, but it belongs to an oracle that has already been told the unknown distribution. The genuine self-induced functional is the median and carries the median variance. In a regular parametric location model, a regular estimator cannot reproduce the smaller oracle variance uniformly in a neighbourhood when that variance lies below the information bound. Pointwise superefficiency remains a separate possibility. The oracle value should be read as a geometric benchmark, never as a variance promised by a feasible procedure.
\end{remark}

\subsection{Median inference inherited by the surrogate}

\begin{proposition}[Exact order-statistic interval]\label{prop:median-exact-ci}
Let $B\sim\operatorname{Bin}(n,1/2)$ and choose an integer $k$ such that
\[
\Prob(k\leq B\leq n-k)\geq1-\alpha.
\]
Then
\[
\bigl[X_{(k)},X_{(n-k+1)}\bigr]
\]
is a distribution-free confidence interval for the intrinsic median $m$ with coverage at least $1-\alpha$.
\end{proposition}

A density-based asymptotic standard error is
\[
\widehat{\operatorname{se}}(\widehat b_n^{\mathrm{int}})
=
\frac1{2\sqrt n\,\widehat f(\widehat b_n^{\mathrm{int}})},
\]
provided $\widehat f$ is consistent at $m$. Under the same local smoothness condition, the ordinary nonparametric bootstrap for the sample median is valid and applies directly to $\widehat b_n^{\mathrm{int}}$.

\section{Inference with estimated charts}\label{sec:estimated}

The estimated chart is where probability geometry becomes fully inferential. A fixed chart specifies in advance which observation is mapped to the probability midpoint $1/2$ and how rapidly coordinates change around that location. When this location and scale calibration is not prescribed by the application, keeping it fixed makes the resulting barycenter depend on arbitrary measurement units. The adaptive construction therefore retains a declared chart shape and learns only its location and scale. The empirical intrinsic surrogate instead estimates the entire CDF and uniformises the observed ranks. The adaptive construction preserves a genuine averaging step while restoring location--scale equivariance.

\subsection{Adaptive probability barycenter}

Let $G_0$ be a fixed chart on $\R$ with density $g_0=G_0'$, and consider
\[
G_\theta(x)
=
G_0\!\left(\frac{x-\mu}{\sigma}\right),
\qquad
\theta=(\mu,\sigma)\in\R\times(0,\infty).
\]
Let $\widehat\theta_n=(\widehat\mu_n,\widehat\sigma_n)$ be preliminary estimators and define
\[
\widehat m_n
=
\frac1n\sum_{i=1}^nG_{\widehat\theta_n}(X_i),
\qquad
\widehat b_n^{\mathrm{ad}}
=
G_{\widehat\theta_n}^{-1}(\widehat m_n).
\]
Equivalently,
\[
\widehat b_n^{\mathrm{ad}}
=
\widehat\mu_n
+
\widehat\sigma_nG_0^{-1}\!\left(
\frac1n\sum_{i=1}^n
G_0\!\left(\frac{X_i-\widehat\mu_n}{\widehat\sigma_n}\right)
\right).
\]

\begin{proposition}[Location--scale equivariance]\label{prop:adaptive-equivariance}
If $\widehat\mu_n$ and $\widehat\sigma_n$ are location--scale equivariant, then
\[
\widehat b_n^{\mathrm{ad}}(aX_1+c,\ldots,aX_n+c)
=
a\widehat b_n^{\mathrm{ad}}(X_1,\ldots,X_n)+c
\]
for all $a>0$ and $c\in\R$.
\end{proposition}

\begin{proof}
Let $Y_i=aX_i+c$ with $a>0$. Equivariance of the preliminary estimators gives
\[
\widehat\mu_n(Y)=a\widehat\mu_n(X)+c,
\qquad
\widehat\sigma_n(Y)=a\widehat\sigma_n(X).
\]
Consequently, the standardised observations are unchanged:
\[
\frac{Y_i-\widehat\mu_n(Y)}{\widehat\sigma_n(Y)}
=
\frac{X_i-\widehat\mu_n(X)}{\widehat\sigma_n(X)}.
\]
Substitution into the explicit formula for $\widehat b_n^{\mathrm{ad}}$ proves the result.
\end{proof}

\subsection{Asymptotically linear expansion}

Let $\theta_0$ be the probability limit of $\widehat\theta_n$, and define
\[
m(\theta)=\E[G_\theta(X)],
\qquad
b(\theta)=G_\theta^{-1}(m(\theta)),
\qquad
b_0=b(\theta_0).
\]

\begin{assumption}[Estimated-chart regularity]\label{ass:adaptive}
The following conditions hold.
\begin{enumerate}[label=(\roman*)]
\item $\widehat\theta_n$ is asymptotically linear:
\[
\sqrt n(\widehat\theta_n-\theta_0)
=
\frac1{\sqrt n}\sum_{i=1}^n\psi_\theta(X_i)+o_{\Prob}(1),
\]
where $\E[\psi_\theta(X)]=0$ and $\E\|\psi_\theta(X)\|^2<\infty$.
\item The class $\{G_\theta:\theta\in\Theta_0\}$ is $P$-Donsker on a neighbourhood $\Theta_0$ of $\theta_0$, and $\theta\mapsto G_\theta$ is differentiable at $\theta_0$ in $L^2(P)$.
\item $m(\theta)$ is differentiable at $\theta_0$, $G_{\theta_0}$ is differentiable at $b_0$, and $G'_{\theta_0}(b_0)>0$.
\end{enumerate}
\end{assumption}

\begin{remark}[Location--scale sufficient conditions]
Assumption~\ref{ass:adaptive} states the empirical-process requirement directly, because boundedness by itself does not justify stochastic equicontinuity. For the location--scale family above, the Donsker property follows from the finite-dimensional monotone structure of the CDF class under standard measurability conditions. Differentiation under the expectation also requires an integrable envelope for
\[
\nabla_\theta G_\theta(x)
=
\left(
-\frac1\sigma g_0(z),
-\frac z\sigma g_0(z)
\right),
\qquad z=\frac{x-\mu}{\sigma}.
\]
Boundedness of $g_0$ alone does not control the scale derivative. The term $|z|g_0(z)$ also needs an integrable envelope.
\end{remark}

\begin{theorem}[Adaptive asymptotic linearity]\label{thm:adaptive-al}
Under Assumption~\ref{ass:adaptive},
\[
\sqrt n(\widehat b_n^{\mathrm{ad}}-b_0)
=
\frac1{\sqrt n}\sum_{i=1}^n\IF_{\mathrm{ad}}(X_i)+o_{\Prob}(1),
\]
where
\[
\IF_{\mathrm{ad}}(x)
=
\frac{
G_{\theta_0}(x)-m(\theta_0)
+
\Delta^\top\psi_\theta(x)
}
{G'_{\theta_0}(b_0)},
\]
and
\[
\Delta
=
\nabla_\theta m(\theta_0)
-
\nabla_\theta G_\theta(b_0)\big|_{\theta_0}.
\]
It follows that
\[
\sqrt n(\widehat b_n^{\mathrm{ad}}-b_0)
\rightsquigarrow
\mathcal N(0,\tau_{\mathrm{ad}}^2),
\qquad
\tau_{\mathrm{ad}}^2=\E[\IF_{\mathrm{ad}}(X)^2].
\]
\end{theorem}

\begin{proof}
Let $W_n(\theta)=n^{-1}\sum_iG_\theta(X_i)$ and $\Phi(w,\theta)=G_\theta^{-1}(w)$. Stochastic equicontinuity and differentiability give
\[
\begin{aligned}
\sqrt n\{W_n(\widehat\theta_n)-m(\theta_0)\}
={}&
\frac1{\sqrt n}\sum_{i=1}^n
\{G_{\theta_0}(X_i)-m(\theta_0)\}\\
&+\nabla_\theta m(\theta_0)^\top
\sqrt n(\widehat\theta_n-\theta_0)
+o_{\Prob}(1).
\end{aligned}
\]
Differentiating $G_\theta(\Phi(w,\theta))=w$ gives
\[
\partial_w\Phi=\frac1{G'_{\theta_0}(b_0)},
\qquad
\nabla_\theta\Phi
=-\frac{\nabla_\theta G_\theta(b_0)|_{\theta_0}}{G'_{\theta_0}(b_0)}.
\]
The delta method and the asymptotic linearity of $\widehat\theta_n$ yield the stated expansion.
\end{proof}

\begin{corollary}[Calibration orthogonality]\label{cor:orthogonality}
If $\Delta=0$, chart calibration has no first-order effect and
\[
\tau_{\mathrm{ad}}^2
=
\frac{\Var(G_{\theta_0}(X))}{[G'_{\theta_0}(b_0)]^2}.
\]
\end{corollary}

\begin{remark}[Reading the correction term]
The vector $\Delta$ measures a mismatch. It compares the average response of the chart to an infinitesimal change in calibration with the response of the chart evaluated at the barycenter itself. When those responses agree, calibration disappears at first order. Such agreement is a structural property of the chart--law pair, not a generic reward for using robust preliminary estimators.
\end{remark}

\subsection{Implementable variance estimation}

An asymptotic variance that cannot be estimated is only half an inferential result. Suppose that estimated influence contributions $\widehat\psi_{\theta,i}$ are available for the preliminary estimator. Define
\[
\widehat\Delta
=
\frac1n\sum_{i=1}^n
\nabla_\theta G_{\widehat\theta_n}(X_i)
-
\nabla_\theta G_\theta(\widehat b_n^{\mathrm{ad}})
\big|_{\theta=\widehat\theta_n},
\]
and
\[
\widehat{\IF}_{i}
=
\frac{
G_{\widehat\theta_n}(X_i)-\widehat m_n
+
\widehat\Delta^\top\widehat\psi_{\theta,i}
}
{G'_{\widehat\theta_n}(\widehat b_n^{\mathrm{ad}})}.
\]
Let
\[
\overline{\widehat{\IF}}
=
\frac1n\sum_{i=1}^n\widehat{\IF}_{i},
\qquad
\widehat\tau_{\mathrm{ad}}^2
=
\frac1{n-1}\sum_{i=1}^n
(\widehat{\IF}_{i}-\overline{\widehat{\IF}})^2.
\]

\begin{proposition}[Consistency of the adaptive sandwich variance]\label{prop:adaptive-var}
Assume Theorem~\ref{thm:adaptive-al}, $L^2(P)$-consistent estimation of $\psi_\theta$, and consistent estimation of the derivative terms entering $\Delta$ and $G'_{\theta_0}(b_0)$. Then
\[
\widehat\tau_{\mathrm{ad}}^2\xrightarrow{\Prob}\tau_{\mathrm{ad}}^2.
\]
\end{proposition}

\begin{remark}[Z-estimator calibration]
If $\widehat\theta_n$ solves
\[
\frac1n\sum_{i=1}^n\varphi(X_i,\widehat\theta_n)=0,
\]
with nonsingular
\[
A=\E\bigl[\partial_\theta\varphi(X,\theta_0)\bigr],
\]
then
\[
\psi_\theta(x)=-A^{-1}\varphi(x,\theta_0).
\]
This gives the implementable contribution
\[
\widehat\psi_{\theta,i}
=-\widehat A^{-1}\varphi(X_i,\widehat\theta_n).
\]
\end{remark}

\begin{remark}[Raw median--MAD calibration]\label{rem:median-mad}
The numerical illustration uses the raw MAD. Let $\mu_0$ be the unique median of $P$ and let $\sigma_0$ be the unique median of $|X-\mu_0|$. Suppose that $P$ has a density $f$ that is continuous and positive at $\mu_0$ and at $\mu_0\pm\sigma_0$. The influence representation of the preliminary estimator $\theta=(\mu,\sigma)$ is $\psi_\theta=(\psi_\mu,\psi_\sigma)^\top$, where
\[
\psi_\mu(x)
=
\frac{1/2-\mathbf 1_{(-\infty,\mu_0]}(x)}{f(\mu_0)}
\]
and
\[
\psi_\sigma(x)
=
\frac{
1/2-\mathbf 1_{[\mu_0-\sigma_0,\,\mu_0+\sigma_0]}(x)
-\{f(\mu_0+\sigma_0)-f(\mu_0-\sigma_0)\}\psi_\mu(x)
}{f(\mu_0+\sigma_0)+f(\mu_0-\sigma_0)}.
\]
For a distribution symmetric about $\mu_0$, the density-difference term vanishes. A consistency-corrected MAD multiplies both the scale functional and its influence function by the chosen correction constant and therefore defines a different calibrated target.
\end{remark}

\begin{corollary}[Adaptive Wald interval]\label{cor:adaptive-wald}
Under Proposition~\ref{prop:adaptive-var},
\[
\widehat b_n^{\mathrm{ad}}
\pm
z_{1-\alpha/2}
\frac{\widehat\tau_{\mathrm{ad}}}{\sqrt n}
\]
is an asymptotic $(1-\alpha)$ confidence interval for $b_0$.
\end{corollary}

\subsection{Bootstrap with chart recalibration}

Calibration must travel with the resample. For each bootstrap sample, recompute the preliminary estimator, rebuild the chart, and only then calculate the barycenter:
\[
\widehat\theta_n^*
\longrightarrow
G_{\widehat\theta_n^*}
\longrightarrow
\widehat b_n^{\mathrm{ad},*}.
\]
Resampling the transformed observations under a frozen $\widehat\theta_n$ answers a different question. It suppresses the calibration term and its covariance with the coordinate fluctuation. Unless calibration is first-order orthogonal, the frozen-chart bootstrap generally estimates the wrong first-order law. Because the covariance contribution has no fixed sign, the resulting uncertainty may be understated or overstated.

\begin{theorem}[Bootstrap validity for the adaptive barycenter]\label{thm:adaptive-bootstrap}
Assume Theorem~\ref{thm:adaptive-al} and bootstrap stochastic equicontinuity of the chart class. Let $N_1^*,\ldots,N_n^*$ be the multinomial counts of the ordinary nonparametric bootstrap sample, and suppose that the preliminary estimator satisfies the conditional expansion
\[
\sqrt n(\widehat\theta_n^*-\widehat\theta_n)
=
\frac1{\sqrt n}\sum_{i=1}^n
(N_i^*-1)\psi_\theta(X_i)
+o_{\Prob^*}(1)
\]
in probability. Then, conditionally on the data,
\[
\sqrt n\left(
\widehat b_n^{\mathrm{ad},*}-\widehat b_n^{\mathrm{ad}}
\right)
\rightsquigarrow
\mathcal N(0,\tau_{\mathrm{ad}}^2)
\]
in probability.
\end{theorem}

\begin{proof}
Write
\[
\mathbb G_n^*h
=
\frac1{\sqrt n}\sum_{i=1}^n(N_i^*-1)h(X_i).
\]
The conditional expansion of $\widehat\theta_n^*$ and bootstrap stochastic equicontinuity give
\[
\sqrt n\{W_n^*(\widehat\theta_n^*)-W_n(\widehat\theta_n)\}
=
\mathbb G_n^*G_{\theta_0}
+\nabla_\theta m(\theta_0)^\top
\mathbb G_n^*\psi_\theta
+o_{\Prob^*}(1)
\]
in probability. Applying the same derivative of $\Phi(w,\theta)=G_\theta^{-1}(w)$ as in the proof of Theorem~\ref{thm:adaptive-al} yields
\[
\sqrt n(\widehat b_n^{\mathrm{ad},*}-\widehat b_n^{\mathrm{ad}})
=
\mathbb G_n^*\IF_{\mathrm{ad}}+o_{\Prob^*}(1).
\]
Since $\IF_{\mathrm{ad}}\in L^2(P)$, the conditional bootstrap central limit theorem completes the proof.
\end{proof}

The theorem supports percentile and basic bootstrap intervals. When $\widehat\tau_{\mathrm{ad}}^*$ is recomputed within each bootstrap sample, the studentised statistic
\[
T_n^*
=
\frac{\sqrt n(
\widehat b_n^{\mathrm{ad},*}-\widehat b_n^{\mathrm{ad}})}
{\widehat\tau_{\mathrm{ad}}^*}
\]
provides a studentised bootstrap interval and is the preferred implementation when sample size permits. A multiplier bootstrap based directly on the estimated influence contributions $\widehat{\IF}_i$ gives a cheaper first-order alternative.

\section{Joint inference for initial Kolmogorov and centred coordinate moments}\label{sec:moments}

A moment sequence is a connected statistical object. Estimating each order separately conceals the dependence created by the common transformed sample. For a fixed integer $q\geq1$, we begin with the whole coordinate vector and define
\[
\mathbf p_q
=
(p_1,\ldots,p_q)^\top,
\qquad
p_r=\E[U^r],
\]
where the superscript $(G)$ is suppressed for notational economy, with empirical counterpart
\[
\widehat{\mathbf p}_q
=
(\widehat p_1,\ldots,\widehat p_q)^\top,
\qquad
\widehat p_r
=
\frac1n\sum_{i=1}^nU_i^r.
\]

\subsection{Initial coordinate moments}

\begin{theorem}[Joint central limit theorem]\label{thm:joint-moments}
For fixed $q$,
\[
\sqrt n(\widehat{\mathbf p}_q-\mathbf p_q)
\rightsquigarrow
\mathcal N_q(0,\Sigma_q),
\]
where
\[
(\Sigma_q)_{rs}
=
\Cov(U^r,U^s)
=
p_{r+s}-p_rp_s,
\qquad 1\leq r,s\leq q.
\]
\end{theorem}

\begin{proof}
The vector $(U,U^2,\ldots,U^q)$ is bounded. The result is the multivariate central limit theorem.
\end{proof}

Define
\[
(\widehat\Sigma_q)_{rs}
=
\frac1{n-1}\sum_{i=1}^n
(U_i^r-\widehat p_r)(U_i^s-\widehat p_s).
\]

\begin{proposition}[Consistent covariance estimator]\label{prop:moment-cov}
For fixed $q$,
\[
\widehat\Sigma_q\xrightarrow{\Prob}\Sigma_q.
\]
\end{proposition}

Whenever $\Sigma_q$ is nonsingular,
\[
n(\widehat{\mathbf p}_q-\mathbf p_q)^\top
\widehat\Sigma_q^{-1}
(\widehat{\mathbf p}_q-\mathbf p_q)
\rightsquigarrow
\chi_q^2,
\]
which yields a joint Wald ellipsoid for the first $q$ coordinate moments.

\subsection{Initial Kolmogorov moments}

Let
\[
\mathbf M_q
=
(M_1^{(G)},\ldots,M_q^{(G)})^\top,
\qquad
M_r^{(G)}=G^{-1}(p_r),
\]
and define
\[
D_q
=
\diag\left(
\frac1{G'(M_1^{(G)})},\ldots,
\frac1{G'(M_q^{(G)})}
\right).
\]

\begin{corollary}[Joint inference for initial Kolmogorov moments]\label{cor:pulled-moments}
If $G'$ is continuous and positive at $M_r^{(G)}$ for $r=1,\ldots,q$, then
\[
\sqrt n(\widehat{\mathbf M}_q-\mathbf M_q)
\rightsquigarrow
\mathcal N_q(0,D_q\Sigma_qD_q),
\]
where
\[
\widehat M_r^{(G)}=G^{-1}(\widehat p_r).
\]
A consistent covariance estimator is
\[
\widehat D_q\widehat\Sigma_q\widehat D_q,
\qquad
\widehat D_q
=
\diag\left(
\frac1{G'(\widehat M_1^{(G)})},\ldots,
\frac1{G'(\widehat M_q^{(G)})}
\right).
\]
\end{corollary}

For an individual initial Kolmogorov moment, the interval should be constructed for $p_r$ and transported through $G^{-1}$. This avoids imposing artificial symmetry in value space.

\subsection{Centred coordinate moments}

Set $p_0=1$. For $r\geq2$,
\[
\mu_{r,c}^{(G)}
=
\sum_{j=0}^r
\binom rj(-p_1)^{r-j}p_j.
\]
The centred coordinate-moment vector is a polynomial image of the initial coordinate-moment vector. No separate probability theory is needed. Once the covariance of $(\widehat p_1,\ldots,\widehat p_q)$ has been estimated, the dependence among centred dispersion and shape summaries follows from the Jacobian of this map.

\begin{corollary}[Joint inference for centred coordinate moments]\label{cor:centered-moments}
Let $h_q$ denote the polynomial map from $(p_1,\ldots,p_q)$ to the selected centred coordinate moments, and let $H_q$ be its Jacobian at $\mathbf p_q$. Then
\[
\sqrt n\{h_q(\widehat{\mathbf p}_q)-h_q(\mathbf p_q)\}
\rightsquigarrow
\mathcal N(0,H_q\Sigma_qH_q^\top),
\]
with consistent plug-in covariance estimator
\[
\widehat H_q\widehat\Sigma_q\widehat H_q^\top.
\]
\end{corollary}

\subsection{Simultaneous concentration intervals}

Because every $U^r\in[0,1]$, Hoeffding's inequality and the union bound yield a simultaneous finite-sample result.

\begin{proposition}[Simultaneous distribution-free intervals]\label{prop:moment-hoeffding}
Let
\[
\varepsilon_{n,\alpha,q}
=
\sqrt{\frac{\log(2q/\alpha)}{2n}}.
\]
Then
\[
\Prob\left(
|\widehat p_r-p_r|\leq\varepsilon_{n,\alpha,q}
\text{ for all }r=1,\ldots,q
\right)
\geq1-\alpha.
\]
Thus the intervals $[\widehat p_r-\varepsilon_{n,\alpha,q},\widehat p_r+\varepsilon_{n,\alpha,q}]\cap[0,1]$ are simultaneous confidence intervals for the first $q$ coordinate moments. Pulling these intervals through $G^{-1}$ gives simultaneous confidence intervals for the initial Kolmogorov moments $M_1^{(G)},\ldots,M_q^{(G)}$.
\end{proposition}

\subsection{Bootstrap inference for moment vectors}

The bootstrap resamples observations, not moment orders. Each selected observation carries the complete vector
\[
(U_i,U_i^2,\ldots,U_i^q),
\]
so the dependence among the estimated moments is preserved automatically.

\begin{theorem}[Bootstrap validity for moment vectors]\label{thm:moment-bootstrap}
For fixed $q$, the ordinary nonparametric bootstrap consistently estimates the joint law of
\[
\sqrt n(\widehat{\mathbf p}_q-\mathbf p_q).
\]
Under the conditions of Corollaries~\ref{cor:pulled-moments} and~\ref{cor:centered-moments}, the bootstrap delta method also validates the initial Kolmogorov moments and the centred coordinate moments.
\end{theorem}

\section{Robustness and efficiency}\label{sec:robustness}

Bounded probability coordinates are often credited with robustness. That description is incomplete unless the inverse chart is kept in view. Boundedness controls the numerator of the influence function. The derivative of the chart at the barycenter controls the amplification that occurs when the coordinate perturbation is pulled back to value space.

\subsection{Fixed-chart influence function}

For a fixed chart,
\[
T_G(P)=G^{-1}\!\bigl(\E_P[G(X)]\bigr)
\]
has influence function
\[
\IF(x;T_G,P)
=
\frac{G(x)-p_G}{G'(b_G)}.
\]
It is bounded because $G(x)\in(0,1)$, provided $G'(b_G)>0$. The gross-error sensitivity is
\[
\gamma^*(T_G,P)
=
\frac{\max(p_G,1-p_G)}{G'(b_G)},
\]
where the value is interpreted as a supremum if the chart does not attain $0$ or $1$ on $I$.

\subsection{Maximal bias under contamination}

Let $P_\varepsilon=(1-\varepsilon)P+\varepsilon Q$. Since $\E_Q[G(X)]\in(0,1)$, the closure of the contaminated functional range is
\[
\left[
G^{-1}\!\bigl((1-\varepsilon)p_G\bigr),
G^{-1}\!\bigl((1-\varepsilon)p_G+\varepsilon\bigr)
\right].
\]
The endpoints need not be attained by a law on $I$, but they give the sharp closure bound. The finite bound for $\varepsilon<1$ reflects anchoring and should not be compared directly with the breakdown point of an equivariant location estimator.

\subsection{Adaptive robustness}

The adaptive influence function from Theorem~\ref{thm:adaptive-al} is bounded whenever the chart coordinate is bounded, $G'_{\theta_0}(b_0)>0$, and the preliminary influence function $\psi_\theta$ is bounded. Median--MAD calibration is a natural robust choice.

\subsection{Reading inferential uncertainty through the chart}

For a fixed chart,
\[
\tau_G^2
=
\frac{\Var(G(X))}{[G'(b_G)]^2}
\]
separates coordinate variability from inverse-chart amplification. For an estimated chart, the numerator becomes
\[
\Var\left(
G_{\theta_0}(X)-m(\theta_0)
+
\Delta^\top\psi_\theta(X)
\right),
\]
which incorporates the calibration contribution and its covariance with the coordinate mean. Relative to the frozen-chart variance, this change is not sign-definite.

The first term, coordinate variability, records the dispersion that remains after the observations have been placed on $(0,1)$. The factor $1/G'(b_G)^2$ is local geometric amplification. A flat chart at the barycenter produces a steep inverse and enlarges both the influence curve and the sampling variance. Estimated charts incorporate the calibration contribution $\Delta^\top\psi_\theta(X)$ together with its covariance with the coordinate mean.

Tail saturation and local amplification act in different parts of the formula. Saturation limits the leverage of observations far into the tails. A small derivative at the barycenter offers no such protection. It magnifies uncertainty exactly where the estimator is pulled back.

\subsection{Relation to M-estimators}

Adaptive chart barycenters belong to the broad family of asymptotically linear robust location estimators, yet their construction has a recognisable geometric character. Classical M-estimation modifies a score or objective on the observation scale. Chart estimation first changes the scale on which arithmetic averaging is performed. The shape $G_0$ plays a role similar to a tuning function, while the calibration rule determines equivariance and contributes part of the final influence curve.

\section{Numerical illustrations}\label{sec:numerics}

The following calculations are deliberately focused. They illustrate the distinction between an anchored fixed chart and an equivariant calibrated chart. They are not presented as a comprehensive comparison of confidence-interval procedures.

\subsection{Asymptotic variance benchmarks}

For symmetric distributions centred at zero, Table~\ref{tab:avar} reports fixed-chart asymptotic variances
\[
\frac{\Var(G(X))}{[G'(0)]^2}
\]
together with classical comparators. Fixed-chart rows are anchored at the true centre and should be read as oracle-style benchmarks for their own chart-dependent targets. They are not feasible equivariant location estimators.

\begin{table}[H]
\centering
\caption{Asymptotic variance benchmarks for symmetric models centred at $0$.}
\label{tab:avar}
\begin{tabular}{lccc}
\toprule
Estimator & Cauchy & Student $t_2$ & $\mathcal N(0,1)$\\
\midrule
Sample mean & $\infty$ & $\infty$ & $1.000$\\
Sample median & $2.467$ & $2.000$ & $1.571$\\
Huber ($k=1.345$) & $2.842$ & $1.846$ & $1.053$\\
Fixed logistic chart & $1.471$ & $1.104$ & $0.694$\\
Fixed Gaussian chart & $0.797$ & $0.678$ & $0.524$\\
Fixed Cauchy chart & $0.822$ & $0.636$ & $0.450$\\
\midrule
Frozen intrinsic oracle & $0.822$ & $0.667$ & $0.524$\\
\bottomrule
\end{tabular}
\end{table}

\subsection{Monte Carlo variance comparison}

Table~\ref{tab:sim} reports $n$ times the Monte Carlo variance at $n=100$ over $20{,}000$ replications. The contaminated Gaussian law is $0.9\mathcal N(0,1)+0.1\mathcal N(0,10^2)$. The adaptive logistic estimator uses the sample median and raw MAD for calibration. The Huber estimator uses $k=1.345$ and the normal-consistency-corrected MAD as a fixed scale within each sample.

The Cauchy sample mean is labelled divergent because its population variance is infinite and a finite replication estimate has no stable target. The fixed-chart values agree with their asymptotic benchmarks, while the gap between fixed and adaptive logistic rows displays the effect of calibration in these models.

The simulation uses NumPy's PCG64 generator with seed $20260902$. The source archive contains the script that produces Table~\ref{tab:sim}, including the contamination mechanism and the numerical definition of the Huber estimate.

\begin{table}[H]
\centering
\caption{Simulated $n\cdot\Var$ at $n=100$ over $20{,}000$ replications.}
\label{tab:sim}
\footnotesize
\setlength{\tabcolsep}{4pt}
\begin{tabular}{lccc}
\toprule
Estimator & Cauchy & Contaminated Gaussian & $\mathcal N(0,1)$\\
\midrule
Sample mean & divergent & $10.76$ & $0.99$\\
Sample median & $2.54$ & $1.86$ & $1.54$\\
Huber ($k=1.345$, MAD scale) & $3.63$ & $1.45$ & $1.04$\\
Fixed logistic chart & $1.48$ & $0.96$ & $0.69$\\
Adaptive logistic (median/raw MAD) & $3.08$ & $1.53$ & $1.10$\\
\bottomrule
\end{tabular}
\end{table}

\section{Conclusion}

The paper began from a practical discomfort. Once observations have been transported to $(0,1)$, coordinate averages and moments are easy to handle, yet the inferential problem is far from finished. The barycenter and initial Kolmogorov moments must return through $G^{-1}$, centred coordinate moments must be analysed jointly on the probability scale, and the chart may itself have been learned from the data. Probability-coordinate inference has a geometry of uncertainty of its own.

For a fixed chart, the clean route is to estimate and studentise the coordinate mean, form the confidence set on the bounded scale, and pull the entire set back. This gives asymptotic intervals that respect the chart and finite-sample Hoeffding intervals when a distribution-free guarantee is preferred. The intrinsic case gives a useful warning. Replacing the unknown CDF by the empirical CDF does not approximate the frozen oracle mechanism. It uniformises the observed ranks and leaves a median-type order statistic. The oracle influence curve and the median influence curve belong to different functionals.

Estimated charts carry the most interesting inferential burden. Their influence function contains the ordinary coordinate fluctuation and the response of the target to calibration. That second term is measurable, estimable, and unavoidable unless an orthogonality relation happens to remove it. It also explains the correct bootstrap: the chart must be rebuilt inside every resample because calibration is part of the estimator, not a preliminary inconvenience to be frozen and forgotten.

The joint moment theory follows the same philosophy. The finite vector of powers of $G(X)$ should be estimated as a vector, with its covariance intact. Initial Kolmogorov moments follow by pulling the initial coordinate moments through $G^{-1}$. Centred coordinate moments remain on the probability scale because they are not probability levels.

Bounded coordinates do not erase uncertainty. They give it a form that can be read. Coordinate variability records what remains after transformation, the inverse chart determines how strongly that variability returns to value space, and the calibration term records the effect of allowing the geometry to learn from the sample. Once these contributions are displayed separately, the inferential role of the chart is no longer hidden inside a single asymptotic variance.

\end{document}